\documentclass[submission,copyright,creativecommons]{eptcs}
\providecommand{\event}{GandALF 2016}
\usepackage{amsmath}
\usepackage{eps}
\usepackage{amsthm}
\usepackage{amssymb}
\usepackage{complexity}
\usepackage{thmtools}
\let\NE\relax
\let\SC\relax

\usepackage{subcaption}
\newcommand{\States}{\mathsf{States}}
\newcommand{\Final}{\mathtt{F}}
\newcommand{\Agt}{\mathsf{Agt}}
\newcommand{\Act}{\mathsf{Act}}
\newcommand{\Tab}{\mathsf{Tab}}
\newcommand{\Dist}{\mathit{Dist}}
\newcommand{\Allow}{\mathsf{Allow}}
\newcommand{\Supp}{\mathit{Supp}}
\newcommand{\last}{\mathit{last}}
\newcommand{\first}{\mathit{first}}
\newcommand{\NE}{Nash Equilibrium}
\newcommand{\Esp}{\mathbb{E}}
\newcommand{\pr}{\mathbb{P}}
\newcommand{\SC}{\mathrm{SC}}
\newcommand{\exit}{\mathrm{Exit}}
\newcommand{\BR}{\mathrm{BR}}

\newcommand{\argmax}{\mathrm{argmax}}
\newcommand{\comp}[3]{#1\left\langle#2\right\rangle_{-#3}}
\newcommand{\dg}[2]{\widehat{#1}^{#2}}

\def\calG{\mathcal G}
\def\calA{\mathcal A}
\let\epsilon\varepsilon

\newtheorem{tdefinition}{Definition}
\newtheorem{tcorollary}{Corollary}
\newtheorem{tlemma}{Lemma}
\newtheorem{ttheorem}{Theorem}
\newtheorem{tproposition}{Proposition}
\newtheorem{example}{Example}

\newtheorem{remark}{Remark}

\makeatletter
\let\c@tlemma\c@ttheorem
\let\c@tdefinition\c@ttheorem
\let\c@tproposition\c@ttheorem
\let\c@tproposition\c@ttheorem
\let\c@tcorollary\c@ttheorem
\makeatother

\usepackage{tikz}
\usetikzlibrary{calc,shapes,arrows,chains,positioning,automata}

\tikzset{rn/.style={circle,draw, node distance=2cm}}
\tikzset{stoch/.style={fill=black,minimum size=1pt,circle}}
\tikzset{final/.style={rectangle,draw, node distance=2cm}}
\tikzset{annotate/.style={node distance=0.9cm}}
\tikzset{entry/.style={initial by arrow, initial text=, initial where=above}}

\colorlet{drouge}{red}
\colorlet{frouge}{red!20!white}
\colorlet{dbleu}{blue}
\colorlet{fbleu}{blue!30!white}
\colorlet{fbleuc}{blue!20!white}
\colorlet{dviolet}{blue!50!red}
\colorlet{fviolet}{blue!50!red!40!white}
\colorlet{dvert}{green!80!black}
\colorlet{fvert}{green!80!black!20!white}
\colorlet{djaune}{yellow!80!black}
\colorlet{fjaune}{yellow!80!black!20!white}
\colorlet{dgris}{white!60!black}
\colorlet{fgris}{white!90!black}
\colorlet{dgrisf}{white!30!black}
\colorlet{fgrisf}{white!70!black}

\tikzstyle{ptrond}=[draw,circle,minimum height=2mm]
\tikzstyle{ptcarre}=[draw,minimum width=3mm,minimum height=3mm]
\tikzstyle{moyrond}=[draw,circle,minimum height=5mm]
\tikzstyle{moycarre}=[draw,minimum width=4mm,minimum height=4mm]
\tikzstyle{rond}=[draw,circle,minimum height=7mm]
\tikzstyle{carre}=[draw,minimum width=6mm,minimum height=6mm]
\tikzstyle{rouge}=[draw=drouge,fill=frouge]
\tikzstyle{vert}=[draw=dvert,fill=fvert]
\tikzstyle{jaune}=[draw=djaune,fill=fjaune]
\tikzstyle{bleu}=[draw=dbleu,fill=fbleu]
\tikzstyle{violet}=[draw=dviolet,fill=fviolet]
\tikzstyle{gris}=[draw=dgris,fill=fgris]
\tikzstyle{grisf}=[draw=dgrisf,fill=fgrisf]
\tikzstyle{rvert}=[style=rond,style=vert]
\tikzstyle{rrouge}=[style=rond,style=rouge]

\title{Stochastic Equilibria under Imprecise Deviations in
  Terminal-Reward Concurrent Games\thanks{This work is partly
    supported by ERC project EQualIS (308087) and by FP7 project Cassting
    (601148).}}
\def\titlerunning{Equilibria under Imprecise Deviations}

\author{Patricia Bouyer
\qquad
 Nicolas~Markey
\qquad
 Daniel Stan
\institute{LSV,  CNRS \& ENS Cachan, Universit\'e Paris-Saclay, France}
}
\def\authorrunning{P.~Bouyer, N.~Markey \& D.~Stan}

\begin{document}
\maketitle

\begin{abstract}
We study the existence of mixed-strategy equilibria in concurrent games played
on graphs. While existence is guaranteed with safety objectives for each
player, Nash equilibria need not exist when players are given arbitrary
terminal-reward objectives, and their existence is undecidable with
qualitative reachability objectives (and only three players). However, these
results rely on the fact that the players can enforce infinite plays while
trying to improve their payoffs. In~this paper, we~introduce a relaxed notion
of equilibria, where deviations are imprecise.
We~prove that contrary to Nash equilibria, such (stationary) equilibria always
exist, and we develop a \PSPACE\ algorithm to compute~one.

\end{abstract}

\section{Introduction}

Games (especially games played on graphs) are a prominent formalism
for modelling and reasoning about interactions between components of 
computerized systems~\cite{thomas02,henzinger05}. Until recently,
those games have mainly been studied in the special case where only
two players are interacting and have opposite
objectives. This setting is especially relevant for modelling reactive
systems evolving in a presumably hostile environment.
Over the last decade, multi-player games with non-zero-sum objectives
have come into the picture: they~allow for conveniently modelling
complex infrastructures where each individual system tries to fulfill
its own objectives, while still being subject to interactions with the
surrounding systems. As an example, consider (a~simplified version~of)
the team-formation problem~\cite{CKPS10}, as depicted in
Fig.~\ref{fig:team}: several agents are trying to complete tasks; each
task requires some resources, which are shared by the
players. Completing a task thus requires the formation of a team that
has all the required resources for that task: each player selects the
task she wants to achieve (and so proposes her resources for achieving
that task), and if a task receives enough resources, the associated
team receives the corresponding payoff (to~be divided among the
players in the team). In~such a game, there is a need of cooperation
(to~gather enough resources), and an incentive to selfishness
(to~maximise the payoff).

\begin{figure}[tb]
  \centering
  \begin{tikzpicture}
    \draw (0,0) node[draw,circle,minimum width=.6cm] (A) {};
    \draw (-2,-1) node[draw,rounded corners=1.5pt] (B) {$\frac{1}{2},\frac{1}{2}$};
    \draw (2,-1) node[draw,rounded corners=1.5pt](C) {$\vphantom{\frac12}1,0$};
    \draw [-latex'] (A) -- (B) node [midway,above left=-5pt]
    {$\begin{array}{c} A_1\to T_1, A_2\to T_1 \\ A_1\to T_2, A_2\to T_2 \end{array}$};
    \draw [-latex'] (A) -- (C) node [midway,above right=-5pt]
    {$\begin{array}{c} A_1\to T_1, A_2\to T_2 \\ A_1\to T_2, A_2\to T_1 \end{array}$};
    
    \path (7.5,0.2) node {\begin{tabular}[t]{@{}l} player $A_1$ has resources
        $\{r_1,r_2,r_3\}$ \\ player $A_2$ has resources $\{r_2,r_3\}$ \end{tabular}};
    \path (7.5,-.8) node {\begin{tabular}[t]{@{}l} task $T_1$ requires resources
        $\{r_1,r_2\}$ \\ task $T_2$ requires resources $\{r_1,r_3\}$ \end{tabular}};
  \end{tikzpicture}
  \caption{An instance of the team-formation
    problem~\cite{CKPS10}. For any deterministic choice of actions,
    one of the players has an incentive to change her choice: there is
    no pure Nash equilibrium. However there is one mixed Nash
    equilibrium, where each player plays $T_1$ and~$T_2$ uniformly at
    random.}\label{fig:team}
\end{figure}
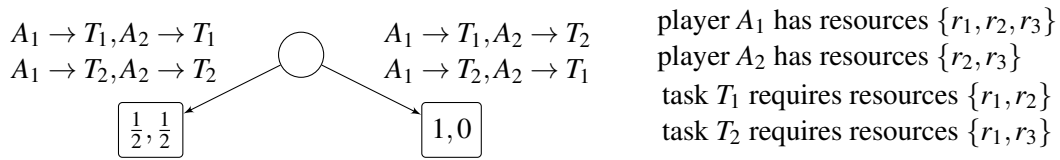

In~that setting, focusing only on optimal strategies for one single
agent is not relevant. In~game theory, several solution concepts have
been defined, which more accurately represents rational behaviours of
these multi-player systems; Nash equilibrium~\cite{nash50} is the
best-known such concept.  A~Nash equilibrium is a strategy profile
(that~is, one~strategy to each player) where no player can improve her
own payoff by unilaterally changing her strategy.  In~other terms, in
a~Nash equilibrium, each individual player has a satisfactory strategy
with regards to the other players'~strategies. Notice that Nash
equilibria need not exist (except for some classes of games) nor be
unique, and they are not necessarily ``optimal'': Nash equilibria
where all players lose may coexist with other Nash equilibria with
positive payoffs. Many other concepts do exist, which refine the
notion of Nash equilibria (like subgame-perfect
equilibria~\cite{selten65} or trembling-hand
equilibria~\cite{selten75}), or relax the notion (like $\epsilon$-Nash
equilibria~\cite{CJM04}).
The existence and computation of (constrained) equilibria (for various
concepts) are important problems in the area, for which many results
have been recently obtained.

In particular, in a recent paper~\cite{BMS14}, we proved that the
existence of Nash equilibria in randomized strategies is undecidable
in deterministic concurrent games with terminal-reward (while the
problem is decidable for pure strategies~\cite{BBMU15}). Those games
are concurrent games played on graphs, with terminal nodes assigning a
reward to every player. The undecidability result holds for three
players or more, and the status of two-player games is open: it is not
known whether there always exists a Nash equilibrium in two-player
concurrent games, even when the terminal rewards are in~$\{0,1\}$ (which
corresponds to a reachability objective).

In order to circumvent this undecidability result, we~consider in this paper a
relaxed version of Nash equilibria, 
with a stronger notion of \emph{profitable deviation}.
A~deviation is
called \emph{really-profitable} only if all the ``neighbouring'' deviations (with
small changes in the probability distribution) remain profitable (in~the
standard sense). 
In this paper, we prove that under this restriction, such equilibria always
exist, even for concurrent games with stochastic states.
We~also show that \emph{stationary} equilibria exist, and provide an algorithm
to compute~one.

To prove the existence result, we show that the notion of imprecise
deviations is captured by adding constraints to the set of
strategy profiles one can~use. This~allows to show the convexity of
the set of best responses to a given strategy profile, as well as a
terminating property (that~is, with a lower-bounded positive
probability the game progresses toward the terminal states). Then
Kakutani's fixed-point theorem~\cite{kakutani1941} can be applied to get
the existence result, as is done in many other contexts. Note that the
above-mentioned terminating property is a property that one either
proves through discounting, like in stay-in-a-set games~\cite{SS01} and
for $\epsilon$-Nash equilibria in reachability games~\cite{CJM04}, or
that one imposes, like in ``games that end almost surely''
in~\cite{AT12}.

\smallskip \noindent
\textbf{Related work.}
Our notion of equilibria is close to the notion of \emph{trembling-hand perfect
equilibria}, which has been proposed in the context of matrix games
in~\cite{selten75}; in trembling-hand equilibria, strategy profiles should
be robust to small perturbations when playing (or implementing) the
strategies while keeping the standard optimality criteria of Nash equilibria.
This concept obviously shares conceptual considerations of our notion of
equilibria against imprecise deviations; however the point-of-view is somehow
dual: the~imprecision is in the implementation of the equilibrium
in~\cite{selten75}, whereas it is in the existence of really-profitable
deviations in our work. While the notion of trembling-hand perfection refines
that of Nash equilibria (it allows for a selection in the set of Nash
equilibria), our notion relaxes that notion. In particular every
trembling-hand perfect equilibrium is a Nash equilibriun, and every Nash
equilibrium is an equilibrium against imprecise deviations (and the inclusions
are strict).

$\epsilon$-Nash equilibria~\cite{CJM04} relax the notion of Nash
equilibria as well, but in a different way: deviations are interpreted
in a standard way, but single deviations should not increase the
payoff by more than $\epsilon$. This is another way to introduce
imprecision in Nash equilibria, which also ensures the existence of
stationary equilibria in stochastic concurrent games with terminal
rewards.

\section{Definitions}
    
In this paper, we study multiplayer stochastic concurrent games. This
section presents a definition of those games, discusses mixed strategy
Nash equilibria, and defines the new concept of equilibria under
imprecise deviations.

\subsection{Concurrent game}

In the following, if $A$ is an at most denumerable set,
$\Dist(A)$ will denote the set of probability distributions over
$A$. If $\delta$ is such a distribution, $\Supp(\delta)$ denotes
the support of $\delta$, that is the subset $\{a\in A \mid
\delta(a) > 0\}$. Pointwise addition for distributions will be written~$+$,
and multiplication by a scalar is written~$\cdot$, so that for any two
distributions~$\delta$  and~$\delta'$ on the same set $A$, and for any $p\in[0,1]$,
$p\cdot \delta+(1-p)\cdot \delta'$ is still a distribution on $A$.

\begin{tdefinition}
  A \emph{
  stochastic concurrent arena}~$\calA$ is a
  $5$-tuple \( \tuple{
    \States,\Agt, \Act, \left(\Allow_i\right)_{i\in\Agt}, \Tab 
  }
  \)
  where
  \begin{itemize}
  \item $\States$ is a finite set of states,
    $\Agt$ is a finite set of agents (or players),
    $\Act$ is a finite set of actions;
  \item for each $i\in\Agt$, $\Allow_i\colon\States \longrightarrow
    2^\Act \backslash\{\emptyset\}$ is a function describing the set
    of actions available to player $i$ from a given state;
  \item $\Tab\colon \States\times \Act^\Agt \rightarrow
    \Dist(\States)$ is the transition function, which assigns to every
    combined action of the players a distribution on the next states.
  \end{itemize}
We say that the arena is \emph{deterministic} whenever the transition
function is deterministic (\textit{i.e.}, only makes use of Dirac
distributions). 
\end{tdefinition}
We~fix a stochastic concurrent
arena~$\calA = \tuple{\States, \Agt, \Act,
  \left(\Allow_i\right)_{i\in\Agt}, \Tab}$ for the rest of this section.
We say a state~$s\in \States$ is \emph{final} if
$\Supp(\Tab(s,A))=\{s\}$ for all $A\in\Act^\Agt$ (that~is, $s$~is a
sink state).  The~set of final states is denoted by~$\Final$.
A \emph{history} (resp. \emph{run}) $\rho$ in $\calA$ is a finite
non-empty (resp. infinite) sequence of states $s_0 s_1 s_2 \dots \in
\States^+$ (resp. $\in \States^\omega$) such that there are actions
$A_1,A_2, \dots \in \Act^\Agt$ with $s_i \in \Tab(s_{i-1},A_i)$ for
every $i \ge 1$. We~denote by~$\first(\rho)$ (resp. $\last(\rho)$,
when relevant) the first (resp. last) state of $\rho$.
We say that $\rho$ is \emph{terminating} whenever it visits a state
in~$\Final$. 

A \emph{reward function} is a function that associates with any
(infinite) run a real number.
This function is \emph{terminal-reward} whenever there exists a function
$\nu\colon \Final \to \mathbb{R}$ such that:
\begin{itemize}
\item any non-terminating run has reward~$0$;
\item if $\rho$ is a terminating run which visits $f \in \Final$, then
  its reward is~$\nu(f)$.
\end{itemize}
In this case, we write $\phi$ as~$\phi_\nu$.
\begin{tdefinition}
  A stochastic concurrent game~$\calG$ is a pair~$\tuple{\calA,
  \phi}$ where $\calA$ is a stochastic concurrent arena and
  $\phi$ associates with each player~$i \in \Agt$ a reward
  function~$\phi_i$. The game has terminal-reward payoffs
  whenever each $\phi_i$ ($i \in \Agt$) is terminal-reward.
\end{tdefinition}

\subsection{Strategies and outcomes}

During a play, players in~$\Agt$ choose their next (distribution over)
moves concurrently and independently of each other, based~on
the current history~$h$ of the play, and what they are allowed to~do
in the current state~$\last(h)$. This is given by strategies, that we
define now.

\begin{tdefinition}
  A \emph{mixed strategy} for player~$i \in \Agt$ is a
  mapping~$\sigma_i \colon \States^+ \to \Dist(\Act)$,
  with the requirement that for all $h\in\States^+$,
  $\Supp(\sigma_i(h)) \subseteq
  \Allow_i(\last(h))$.
\end{tdefinition}

\looseness=-1
Note that strategies, as defined above, can only
observe the sequence of visited states along the history, but they may not
depend on the exact distributions chosen by the players along the
history, nor 
on the actual sequence of actions played by the players.
Notice that this model is more general than the model where actions
are visible, which are sometimes considered in the literature---see
for instance~\cite{ummels08} and \cite[Section 6]{BBMU11} or
\cite{CD14} for discussions---and
the results presented here are valid when considering visible actions.

In this paper, we consider several subclasses of strategies:
\begin{itemize}
\item the set of \emph{mixed strategies} of player~$i$ in arena
  $\mathcal{A}$, denoted~$\mathbb{S}_i^{\mathcal{A}}$, is the set
  containing all the strategies of player~$i$ as defined above;
\item the set of \emph{pure strategies} of player~$i$,
  denoted~$S_i^{\mathcal{A}}$ contains those strategies in which all
  probability distributions are Dirac functions (that is, strategies
  are in some sense deterministic);
\item the set of \emph{stationary strategies} of player~$i$,
  written~$\mathbb{M}_i^{\mathcal{A}}$, in which the value of the
  strategy over history $h$ only depends on $\last(h)$;
\item the set of \emph{(pure) memoryless strategies}, denoted
  with~$M_i^{\mathcal{A}}$, which contains the strategies that are
  pure and stationary.
\end{itemize}

A~strategy profile is a tuple $\sigma = (\sigma_i)_{i \in \Agt}$, in
which $\sigma_i$ is a strategy for player $i$.
Following the definitions introduced above, we consider the full
class~$\mathbb{S}^{\mathcal{A}}$ of mixed strategy profiles, the
class~$S^{\mathcal{A}}$ of pure strategy profiles, the
class~$\mathbb{M}^{\mathcal{A}}$ of stationary strategy profiles, and
the class~$M^{\mathcal{A}}$ of memoryless (that is, pure stationary)
strategy profiles. If $\calA$ is clear in the context, we will
simplify the various notations and skip the superscript $\calA$ in the
notation.

Let $\sigma$ be a strategy profile. We denote by~$\pr^\sigma(-)$
the probability measure induced by~$\sigma$ on the infinite runs
in~$\States^\omega$ as follows: the probability of cylinder $h
\States^\omega$, with $h = s_1 \dots s_p$, is defined as $\pr^\sigma(h
\cdot \States^\omega) = \prod_{i=1}^p \sigma(h_{<i})(s_i)$, where
$h_{<i}$ is the prefix of length $i-1$ of $h$ (if $i=1$, $h_{<i}$ is
the empty word); it extends in a unique way to the $\sigma$-algebra
generated by the above cylinders.

If $h\in\States^+$ is a history such that $\pr^\sigma(h \cdot
\States^\omega)>0$, we define the conditional probability measure
$\pr^\sigma(- \mid h)$ in a natural way: $\pr^\sigma(h' \cdot
\States^\omega\mid h) = \frac{\pr^\sigma(h' \cdot \States^\omega)}{
  \pr^\sigma(h \cdot \States^\omega)}$ if $h$ is a prefix of $h'$ and
$\pr^\sigma(h' \cdot \States^\omega\mid h) = 0$ otherwise; this
extends in a natural way to the generated $\sigma$-algebra.  For any
finite history $h'\in\States^+$, we write $\pr^\sigma(h' \mid h)$ as a
shorthand for $\pr^\sigma(h' \cdot \States^\omega \mid h)$.

For every $i \in \Agt$, let $\phi_i$ be a terminal-reward reward
function for player $i$, and define $\phi = (\phi_i)_{i \in \Agt}$. We
denote by $\Esp^{\sigma}(\phi_i \mid h)$ the expected value of the
reward function $\phi_i$ induced by the probability mesure
$\pr^\sigma(- \mid h)$. By extension, we write $\Esp^{\sigma}(\phi
\mid h)$ for the tuple $(\Esp^{\sigma}(\phi_i \mid h))_{i \in \Agt}$

\subsection{Nash equilibria}
We now define the notion of \NE{}, as introduced by
Nash~\cite{nash50}.
\begin{tdefinition}
  A \emph{\NE{}} from state $s_0$ is a (mixed) strategy profile
  $\sigma \in \mathbb{S}$ such that:
  \[
  \forall i\in\Agt~\forall \sigma'_i \in \mathbb{S}_i \quad
  \Esp^{\sigma[i/ \sigma'_i]}(\phi_i\mid s_0) \leq
  \Esp^{\sigma}(\phi_i\mid s_0).
  \]
  where $\sigma[i/\sigma'_i]$ 
  is the strategy profile obtained from~$\sigma$ by replacing
  strategy~$\sigma_i$ for player~$i$ with~$\sigma'_i$.
  
\end{tdefinition}
In this definition, strategy $\sigma'_i$ corresponds to a
\emph{deviation} of player $i$ with respect to the profile $\sigma$;
we will often use this terminology thereafter.

\begin{figure}[tb]
\begin{minipage}{.7\linewidth}
  \begin{subfigure}{0.35\textwidth}
  \centering
  \begin{tikzpicture}
    \node (win) [final] {$\vphantom{\frac13}-1,1$};
    \node (lose) [final, right of=win, node distance=2.3cm] {$\vphantom{\frac13}1,-1$};
    \node[minimum width=7mm] (s) at ($(win)!.5!(lose) + (0,1.3cm)$) [rn] {};

    \path (s.90) node[above=0.3cm] (entry) {}
                 edge[-latex',draw] (s);

    \path (s) edge[-latex'] node[left,pos=.4] {$sh$, $wr$} (win)
              edge[-latex'] node[right,pos=.4] {$sr$} (lose)
              edge[-latex',loop right, looseness=7,out=-20,in=40] node[right]
              {$wh$} (s);

  \end{tikzpicture}
  \caption{Hide-or-run game}
  \label{exHoR}
  \end{subfigure}
  \begin{subfigure}{0.6\textwidth}
  \centering
  \begin{tikzpicture}
    \node (win) [final] {$1,\frac1{3}$};
    \node (lose) [final, left of=win, node distance=3cm] {$\frac1{3},1$};
    \node[minimum width=7mm] (s) at ($(lose) + (0,1.3cm)$) [rn] {$1$};
    \node[minimum width=7mm] (sp) at ($(win) + (0,1.3cm)$) [rn] {$2$};

    \path (s.90) node[above=0.3cm] (entry) {}
                 edge[-latex',draw] (s);

    \path (s) edge[-latex'] node[left,pos=.4] {$s$} (lose)
              edge[-latex', out=10, in=170] node[above,pos=.4] {$c$} (sp);
    \path (sp) edge[-latex'] node[left,pos=.4] {$s$} (win)
               edge[-latex', out=190, in=-10] node[below,pos=.4] {$c$} (s);

  \end{tikzpicture}
  \caption{The first player to quit the loop loses}
  \label{exLF}
  \end{subfigure}
  
\caption{Two examples of games with cycling behaviours}
\label{fig:nonash}
\end{minipage}\hfill
\begin{minipage}{.3\linewidth}
    \begin{tikzpicture}
      \node (win) [final] {$0,0$};
      \node (lose) [final, right of=win, node distance=3cm] {$1,0$};
      \node[minimum width=7mm] (s) at ($(win)!.5!(lose) + (0,1.3cm)$) [rn] {$s$};
      
      \path (s.90) node[above=0.3cm] (entry) {}
      edge[-latex',draw] (s);
      
      \path (s) edge[-latex'] node[left,pos=.4] {$aa$, $ab$, $ba$} (win)
      edge[-latex'] node[right,pos=.4] {$bb$} (lose);
      
    \end{tikzpicture}\par
    $\sigma_1(a \mid s) = 1; \quad \sigma_2(b \mid s)=\epsilon$
\caption{A simple game}\label{fig:nounprecise}
\end{minipage}
\end{figure}

\begin{example}
\label{examples}
Fig.~\ref{fig:nonash} displays two examples of games that we will
describe now.
The hide-or-run game (on the left) represents a game where one player
has one snowball and wants to shoot the other player; the second
player is hiding, and wants to run to the other side of the road. The
first player can either wait or shoot, while the second can hide or
run. Label ``$sr$'' on a transition represents the concurrent action
``$s$ (shoot) for the first player and $r$ (run) for the second
player''.  The~payoff is~$(0,0)$ if the players keep on
playing~''$wh$'' (loop on the initial state).  The first player wins
after ``$sr$'', and loses after ``$sh$'' and ``$wr$'' (represented by
rewards $(1,-1)$ or $(-1,1)$).  One can easily check that this game
has no Nash equilibrium: if the probability to jointly take $wh$
(resp. $sr$) is positive, then the second player can deviate and earn
more with action pair $wr$ (resp. $sh$); if the probability to jointly
take $sh$ (resp. $wr$) is positive, then the first player can deviate
and earn more with action pair $wh$ (resp. $sr$).

The second game is turn-based, and numbers labelling nodes correspond
to the players: in the left-most state, the first player can decide
whether to stop (action $s$) or to continue (action $c$) playing the
game; symmetrically for the second player in the right-most
state. Again, the payoff is~$(0,0)$ if the play does not reach a
terminal state. This game has pure Nash equilibria: for instance, the
memoryless strategy profile where player~$1$ plays~$c$ and player~$2$
plays~$s$ is an equilibrium, with payoff~$(1,1/3)$. Another solution concept
would allow a tradeoff between players who will commit a fixed probability
each to exit the game (for example $\epsilon>0$). In general, such tradeoff
is not a Nash equilibrium as the other player can change his mind (play $c$).
\end{example}

\looseness=-2
While one can compute pure (that is, deterministic) Nash equilibria in
deterministic terminal-reward games~\cite{BBMU15}, in the general
case, computing mixed Nash equilibria in terminal-reward games is
undecidable. Even for turn-based games, \cite{ummels08} proved that it
is impossible to decide wheter a turn-based game with at least $14$
players has a Nash equilibrium where one player wins almost surely
(called $0$-safe condition). This result was later improved by
\cite{WTMKD15} to $0$-safe equibria with finite memory and pure
strategies in turn-based games with at least $5$ players.  In the
concurrent setting, \cite{UW11a} showed the existence of a Nash
equilibrium is undecidable for $14$-player concurrent deterministic
games using similar techniques, and when strategies do not observe the
actions which are played (as in the current paper), the number of
players can even be reduced to $3$ (\cite{BMS14}).  The $0$-safe
condition (one player should win) can be omitted in the concurrent
setting, thanks to a gadget, composed of a $2$-player zero-sum
concurrent game having almost-optimal strategies but no optimal
strategy, hence no Nash equilibrium (this is the first example
mentioned previously, and depicted on Fig.~(\ref{exHoR})).  If only
non-negative terminal rewards are allowed, these undecidability
results still hold in the concurrent setting, but under the additional
$0$-safe condition (there is no known game with no Nash equilibrium in
this setting); indeed, the previous gadgets cannot be adapted as
non-negative terminal rewards imply that every game is non-zero sum,
then no player has an incentive to make the game cycling, ensuring
global payoff $0$, instead of reaching a terminal state.  We summarize
this discussion with the stronger undecidability result which applies
to the precise setting of this paper.

\begin{ttheorem}[\cite{BMS14}]
  The existence problem of a Nash equilibrium in concurrent
  deterministic games with three players and terminal-reward payoff
  functions is undecidable.
\end{ttheorem}

On the positive side, \cite{CJM04} showed that the relaxed notion of
$\varepsilon$-Nash equilibrium, where deviations may only improve the
payoffs by at most $\varepsilon$, always exists and can be
computed.
However, while the game of
  Fig.~(\ref{exLF}) is very symmetric, there is no ($\epsilon$-)Nash
  equilibrium (except the cycling one with payoff $0$ for both
  players) where the two players have close payoffs. This is due to
  the discontinuity yielded by the pure deviation which consists in
  cycling; and if this pure strategy is not played precisely, there
  will actually be no improvement in the payoffs.
  We will therefore propose a new notion of equilibria where
  improvements by deviations should not come from a (punctual)
  discontinuity in the payoff function.

\subsection{Equilibria under imprecise deviations}

In this paper, we propose a new solution concept, with some
\textit{robustness} constraints on possible deviations,
which will enjoy rather nice termination and continuity properties.
\begin{tdefinition}
  An \emph{equilibrium under $\epsilon$-imprecise deviations} from
  state $s_0$ is a strategy profile $\sigma \in \mathbb{S}$ s.t.
  \[\forall i\in\Agt.\
  \forall \sigma'_i \in \mathbb{S}_i.\ \exists \sigma''_i \in
  \mathbb{S}_i\ \text{s.t.}\ \Esp^{\sigma[i/ \sigma''_i]}(\phi_i\mid
  s_0) \leq \Esp^{\sigma}(\phi_i\mid s_0) \ \text{and}\
  d(\sigma'_i,\sigma''_i) \leq \epsilon\]
  where $d(\sigma,\sigma')$ is the supremum distance between the two
  distributions:
  \[
  d(\sigma,\sigma') = \sup_{h\in\States^+} d(\sigma(h), \sigma'(h))
  \]
\end{tdefinition}
The intuition behind that definition is that, to have an incentive to
deviate, a player should be sure to improve her payoff, even if her
deviation is perturbed by $\epsilon$ (this corresponds to some noise
the other players can add, or to a lack of precision in playing
distributions). Said differently, a deviation is only considered
profitable when all the surrounding (up~to a distance of~$\epsilon$)
strategies are also profitable.

We will prove that this new solution concept enjoys very nice
properties: (a) for every $\epsilon>0$, equilibria under
$\epsilon$-imprecise deviations always exist, and (b) we can decide
(and compute) such equilibria with constraints over the payoffs of the
players.

\begin{example}
  Back to the first game in Example~\ref{examples}
  (Fig.~\ref{exHoR}). The strategy profile such that the first player
  plays $s$ with proba $1$ and player $2$ plays $r$ with probability
  $\epsilon$ is an equilibrium under $\epsilon$-imprecise deviations
  with payoff $(2\epsilon-1,1-2\epsilon)$ (only the second player can
  deviate and improve, but its deviation will be smaller (w.r.t. the
  distance) than $\epsilon$).

  In the second game in Example~\ref{examples}
  (Fig.~\ref{exLF}). The strategy profiles where each player plays $s$
  with probability $\epsilon$ yields payoffs ${1-2/(6-3\epsilon)}$ for
  player~$1$ and ${1-(2-2\epsilon)/(6-3\epsilon)}$ for player~$2$ from
  the initial state.  It is an equilibrium under $\epsilon$-imprecise
  deviations. The only way to really improve the payoff for a player
  is to play with higher probability action $c$. But with the lack of
  precision, she might lose some payoff anyway.  The payoff values get
  arbitrarily close to $2/3$ as $\epsilon$ goes to $0$. Such an
  equilibrium is neither a Nash equilibrium, neither a $\epsilon$-Nash
  equilibrium, since the pure deviation $c$ allows an improvement of
  almost $1/3$.

  Finally, consider the game of Fig.~\ref{fig:nounprecise}, and the strategy
  profile~$(\sigma_1,\sigma_2)$: the payoff is then~$(0,0)$, and player~$1$
  can improve her payoff by~$\epsilon$ by playing action~$b$ from~$s$.
  So~$(\sigma_1,\sigma_2)$ is an $\epsilon$-Nash equilibrium but not an
  equilibrium under $\epsilon$-imprecise deviations: any strategy at distance
  $\epsilon$ from~$\sigma'_1$ strictly improves the payoff of player~$1$.
  Thus we conclude that the two concepts are incomparable.
\end{example}

\begin{remark}
  As we already noticed, equilibria under imprecise deviations are not
  Nash equilibria in the classical sense, but Nash equilibria are
  equilibria under imprecise deviations. So our notion relaxes that of
  Nash equilibria.
  Finally the concept of trembling-hand equilibria~\cite{selten75},
  already discussed in the introduction, is an orthogonal notion.
\end{remark}

\section{Existence of equilibria under imprecise deviations}

In this section, we prove the following existence result:

\begin{ttheorem}
  \label{theo:existence}
  Let $\mathcal{G}$ be a stochastic concurrent game with
  terminal-reward payoffs, and let $s_0$ be
  a state of $\mathcal{G}$. For every $\epsilon>0$,
  there always exists an equilibrium under
  $\epsilon$-imprecise deviations from state $s_0$.
\end{ttheorem}

The proof will rely on an alternative notion of equilibria, where
players are enforced to leave cycles of the game. We formalize this
now, and we fix for the rest of this section a stochastic concurrent
game with terminal-reward payoffs $\mathcal{G} = \tuple{\mathcal{A},
  \phi_\nu}$, with $\mathcal{A} =  
\tuple{\States, \Agt, \Act, \left(\Allow_i\right)_{i\in\Agt}, \Tab}$

\subsection{Non-cycling games}
  \begin{tdefinition}
    A~state~$s$ of~$\cal A$ is said \emph{cycling} if there exists a
    mixed strategy profile~$\sigma \in \mathbb{S}$ such that no player
    can enforce (by deviating) reaching a final state, that is:
    \[
    \forall i\in\Agt~\forall \sigma'_i \in \mathbb{S}_i,~
    \pr^{\sigma[i/\sigma'_i]}(\States^*\Final^\omega \mid s) = 0.
    \]
    
    The arena $\mathcal{A}$ (and by extension, the game $\mathcal{G}$)
    is said \emph{cycle-free} if it contains no cycling state.
  \end{tdefinition}
  We notice first that in the above definition, strategy profiles can
  be restricted to memoryless profiles ($\sigma \in M$), and
  deviations can be restricted to stationary deviations ($\sigma'_i
  \in \mathbb{M}_i$). Furthermore only the supports of these
  deviations matter.  

  We further notice that from any cycling state, there is a Nash
  equilibrium with payoff~zero for all the players (playing profile
  $\sigma$ from the definition). Those are also equilibria under
  imprecise deviations (since no payoff can be improved).
  
  They are therefore somehow pathological behaviours, that we will
  remove. This is formalized as follows:

  \begin{restatable}{tproposition}{propcycfree}

    \label{prop:cyclefree}
    One can construct a cycle-free game $\widetilde{\mathcal{G}} =
    \tuple{ \widetilde{\mathcal{A}},\phi_{\widetilde{\nu}}}$
    which has less Nash equilibria and less
    equilibria under imprecise deviations (whatever the bound on the
    imprecision): for every equilibrium (Nash, resp. under imprecise
    deviations) $\widetilde\sigma$ in
    $\widetilde{\calG}$, one can build an equilibrium (Nash, resp.
    under imprecise deviations) $\sigma$ with the same payoffs in $\calG$.
  \end{restatable} 

  This proposition allows to prove Theorem~\ref{theo:existence} by
  restricting to cycle-free games: if the existence holds for
  cycle-free games, then it will hold as well for the whole class of
  stochastic concurrent games with terminal-reward payoffs.

\subsection{Strong components and terminating strategy profiles}

We will see that equilibria under imprecise deviations with stationary
strategies always exist. The main argument of the existence theorem relies
on the structure of the strategy profiles, that can be forced to terminate the
game, even in the presence of deviations. We describe in this subsection
a definition of the constraints we impose on our strategies. These constraints
should be tight enough for the game to terminate, later implying the existence
theorem of a \emph{stable} profile, but should also be general enough for this
same \emph{stable} profile to capture the notion of equilibria under imprecise
deviations.

  \begin{tdefinition}
    \looseness=-1 Let $C$ be a non-empty set of states of $\calA$, and
    $\sigma \in \mathbb{M}$ be a stationary strategy profile. We~say
    that $\sigma$ stabilizes $C$ if for every $s \in C$, for every $s'
    \in \States$, $\pr^\sigma(\States^*\cdot s' \mid s) > 0\quad
    \text{iff}\quad s' \in C$.  When such a profile exists for $C$, we
    say that $C$ is a \emph{strong component}, and write $\SC$ the set
    of strong components.
  \end{tdefinition}
  Notice that for defining the stabilization property, one could
  equivalently require the probability be equal to~$1$. Also notice
  that every strong component intersecting~$\Final$ is reduced to a
  singleton.

  \begin{tdefinition}
    Let $C\in \SC$ be a strong component, and $s\in C$.
    An~action~$a\in\Act$ is an \emph{exiting action} from~$C$ for a
    state~$s$ and player~$i$ if there exists $\sigma \in \mathbb{S}$
    which stabilizes $C$ such that:
    \[
    \pr^{\sigma[i/(s\mapsto a)]} \left(s\cdot\left(\States\setminus
        C\right) \cdot \States^\omega \mid s\right) > 0.
    \]
    We set $\exit(C) = \{(a,i,s) \mid a\ \text{is an exiting action
      from}\ C\ \text{for a state}\ s\ \text{and player}\ i\}$.
  \end{tdefinition}

  We then trivially have:
  \begin{tlemma}
    If $\calA$ is cycle-free, then for any $C\in \SC$, $\exit(C) \neq
    \emptyset$.
  \end{tlemma}
  For the rest of this subsection, we will systematically assume that
  $\calA$ is cycle-free.

  We will now restrict the set of strategy profiles in which we search
  for equilibria. Under this restriction, each play will eventually
  reach a final state with probability~$1$.
  Nash equilibria restricted to this set of strategies will actually
  correspond to our modified notion of equilibria, in a sense that we
  will make precise.

  \begin{tdefinition}
    Let~$\epsilon>0$ and assume 
    $\calA$ is cycle-free.  For every strong component~$C\in \SC$, we
    define the set of \emph{$(\epsilon,C)$-exiting stationary strategy
      profiles} as follows:
    \[
    \Delta_{\epsilon}(C) = \left\{ \sigma\in\mathbb{M} \mid
      \forall (a,i,s)\in\exit(C)
      \sigma_i(s)(a) \geq \epsilon
    \right\}
    \]
    We also let $\Delta_\epsilon = \bigcap_{C\in \SC}\Delta_\epsilon(C)$.
  \end{tdefinition}
  Note that, to be properly defined and non-empty, $\Delta_\epsilon$ requires the
  assumption that the game arena is cycle-free.

  \begin{tlemma}
    For all $\epsilon\leq \frac1{|\Act|}$ and $\calA$ cycle-free,
    it~holds $\Delta_{\epsilon}
    \neq \emptyset$.
  \end{tlemma}
  \begin{proof}
    Consider the stationary strategy profile~$\sigma_u$ which makes
    each player play uniformly at random over the set of allowed
    actions, at each state.

    For~any~$C\in \SC$, since $\exit(C)$ is non-empty, this strategy
    profile is in~$\Delta_\epsilon(C)$. Hence $\sigma_u\in
    \Delta_\epsilon$.
  \end{proof}

The strategy profiles in~$\Delta_\epsilon$ enjoy the following
property, which establishes some kind of fairness with respect to
final states for strategies in $\Delta_\epsilon$. This will be useful
in the sequel:
  \begin{tproposition}
    \label{thmTerm}
    Fix $0<\epsilon\leq \frac1{|\Act|}$ and $\calA$ cycle-free.
    There exist $0<p<1$ and
    $k\in\mathbb{N}$ such that for every $\sigma\in\Delta_\epsilon$,
    for every $s \in \States$, for every $n \ge 0$,
    $\pr^\sigma(\States^{k \cdot n}\cdot\Final^\omega \mid s)\geq
    1-p^n$.
  \end{tproposition}
\subsection{Restricting to memoryless deviations}

This part is devoted to the proof of the following key lemma:
\begin{tlemma}
  Let $s_0$ be a state of a stochastic concurrent game~$\calG$ with
  terminal-reward payoffs.  For any stationary strategy
  profile~$\sigma\in\mathbb{M}$, it~holds: $\sigma$~is an equilibrium
  under $\epsilon$-imprecise deviations iff
      \[\forall i\in\Agt.\
      \forall \sigma'_i \in M_i.\ \exists \sigma''_i \in \mathbb{M}_i.
      \quad d(\sigma'_i,\sigma''_i) \leq \epsilon \wedge
      \Esp^{\sigma[i/ \sigma'_i]}(\phi_i\mid s_0) \leq
      \Esp^{\sigma}(\phi_i\mid s_0) \]
   In other terms, it~is sufficient to consider memoryless
   deviations when checking if a stationary strategy profile is
   an equilibrium under imprecise deviations.
\end{tlemma}

We~prove this lemma by considering an intermediate two-player game to represent
deviations of Player~$i$ and their counter-deviations at distance~$\epsilon$.

The notion of equilibria under imprecise deviation has been introduced in
a very general setting with arbitrarily complex strategies and deviations.
An important step when proving existence of stationary
equilibria is to check that one can restrict ourselves to deviations that
are also stationary. Intuitively, one can even wonder if we can, as in the
case of Nash Equilibria, only consider pure memoryless deviations, that will
be imprecise up to $\varepsilon$, hence leading to stationary deviations, but
in finite number.

Let $\calG$ a game, $\sigma$ a stationary strategy profile and
$i\in\Agt$ a player.  We write~$\comp{\calG}{\sigma}{i}$ for the
$1$-player game obtained from~$\calG$ by assigning to all players, but
player~$i$, her strategy in~$\sigma$.  Note that for any
$\sigma'_i\in\mathbb{S}_i$, we~have
$\Esp_{\calG}^{\sigma[i/\sigma'_i]}(\phi_i\mid s) =
\Esp_{\comp{\calG}{\sigma}{i}}^{\sigma'_i}(\phi_i\mid s)$ In the
following, we are mainly interested in the possible
$\varepsilon$-imprecise deviations of player $i$ alone in this new
game.

\medskip In order to make the reduction clear, we consider in
the following the particular case of games where each player is allowed at
most two actions. When exactly two distinct actions are allowed,
they will be noted $a$ and~$b$. The general case will be discussed in
remark~\ref{rq:dgblow}.

For a stationary profile~$\sigma$, we consider the $1$-player
game~$\comp{\calG}{\sigma}{i}$ as
defined above (with Player~$i$ alone, all other strategies being fixed) and
construct a $2$-player 
turn-based game with an additional antagonistic Player~$\hat{i}$, whose role is
to ``change'' the strategy of Player~$i$ by a distance at most~$\varepsilon$.
Formally, for any state~$s$ where Player~$i$ has two allowed actions~$a$
and~$b$ (resulting in distributions $\delta(s,a)$ and~$\delta(s,b)$, resp.),
we~modify the game as follows:
\makeatletter
\def\low{\@ifnextchar2{[0,2\epsilon]\@gobble}{[0,\epsilon]}}
\def\high{\@ifnextchar2{[1-2\epsilon,1]\@gobble}{[1-\epsilon,1]}}
\makeatother
\begin{itemize}
\item from~$s$, Player~$i$ is given the opportunity to move to one of the
following four states: $(s,\low)$, $(s,\low2)$, $(s,\high2)$ and $(s,\high)$. 
\item from each state $(s,[\alpha,\beta])$, Player~$\hat i$ has two
actions, leading to
distributions~$\alpha\cdot\Tab(s,a)+(1-\alpha)\cdot\Tab(s,b)$ 
and~$\beta\cdot\Tab(s,a)+(1-\beta)\cdot\Tab(s,b)$, respectively.
If~Player~$\hat i$ plays action~$a$ with probability~$p$, then the final
distribution is $[p\alpha+(1-p)\beta]\cdot \Tab(s,a) + 
[(1-p)(1-\alpha)+p(1-\beta)]\cdot \Tab(s,b)$.
\end{itemize}

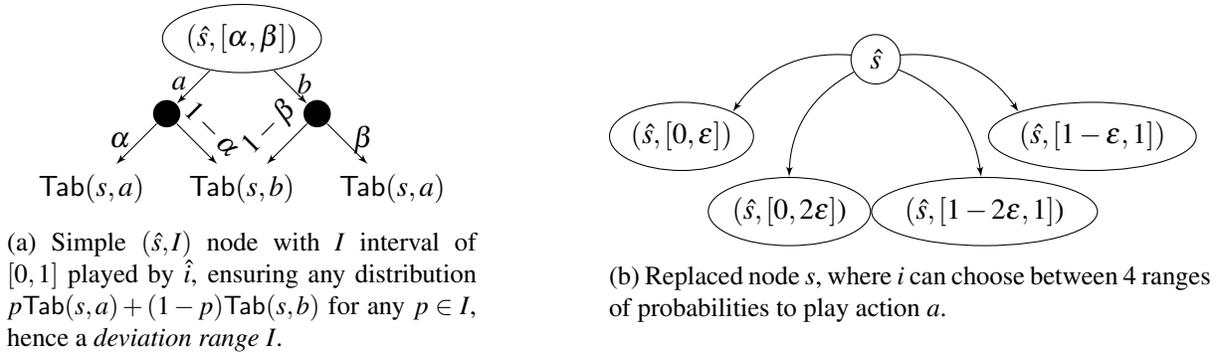
\begin{figure}[bt]
  \begin{subfigure}{0.39\textwidth}
    \centering
    \begin{tikzpicture}
      \node[rn, ellipse, inner xsep=0] (s0) {$(\hat{s},[\alpha,\beta])$};
      \node[below of=s0, node distance=1cm] (mid) {};
      \node[stoch, left of=mid] (beta) {};
      \node[stoch, right of=mid] (alpha) {};
      \node[below of=beta, node distance=1cm] (midbeta) {};
      \node[left of=midbeta] (beta1) {$\Tab(s,a)$};
      \node[below of=alpha, node distance=1cm] (midalpha) {};
      \node[right of=midalpha] (alpha2) {$\Tab(s,a)$};
      \node[] (sb) at ($(beta1)!0.5!(alpha2)$) {$\Tab(s,b)$};
      \path (s0) edge[-latex'] node[left, pos=0.4] {$a$} (beta)
                 edge[-latex'] node[right, pos=0.4] {$b$} (alpha);
      \path (beta) edge[-latex'] node[left, pos=0.4] {$\alpha$} (beta1)
                   edge[-latex'] node[above, sloped] {$1-\alpha$} (sb);
      \path (alpha) edge[-latex'] node[right, pos=0.4] {$\beta$} (alpha2)
                   edge[-latex'] node[above, sloped] {$1-\beta$} (sb);

    \end{tikzpicture}
    \caption{Simple $(\hat{s},I)$ node with $I$ interval of $[0,1]$
        played by $\hat{i}$, ensuring any distribution
        \mbox{$p\Tab(s,a)+(1-p)\Tab(s,b)$}
        for any $p\in I$, hence a \emph{deviation range} $I$.}
  \end{subfigure}\hfill
  \begin{subfigure}{0.5\textwidth}
  \begin{tikzpicture}
    \node[rn] (s0) {$\hat s$};
    \node[below of=s0, node distance=0.7cm] (mid) {};
    \node[left of=mid, node distance=1.9cm] (sa) {};
    \node[rn,ellipse, inner xsep=0,anchor=north east] (s1) at (sa) {$(\hat{s},[0,\varepsilon])$};
    \node[right of=mid, node distance=1.9cm] (sb) {};
    \node[rn,ellipse, inner xsep=0,anchor=north west] (s2) at (sb) {$(\hat{s},[1-\varepsilon,1])$};
    \node[rn,ellipse, inner xsep=0,anchor=north west] (s3) at ($(sa)!0.60!(sb)-(0,1cm)$) {$(\hat{s},[1-2\varepsilon,1]$)};
    \node[rn,ellipse, inner xsep=0,anchor=north east] (s4) at ($(sb)!0.60!(sa)-(0,1cm)$) {$(\hat{s},[0,2\varepsilon]$)};
    \path (s0) edge[-latex', bend right]
    (s1.north east)
    edge[-latex', bend right] 
    (s4)
    edge[-latex', bend left] 
    (s3)
    edge[-latex', bend left] 
    (s2.north west);
  \end{tikzpicture}
  \caption{Replaced node $s$, where $i$ can choose between $4$ ranges of
  probabilities to play action $a$.}
  \end{subfigure}
  \caption{Translation of a node $s$ with allowed action $a$ and $b$ to $\widehat{a}$.}
\end{figure}

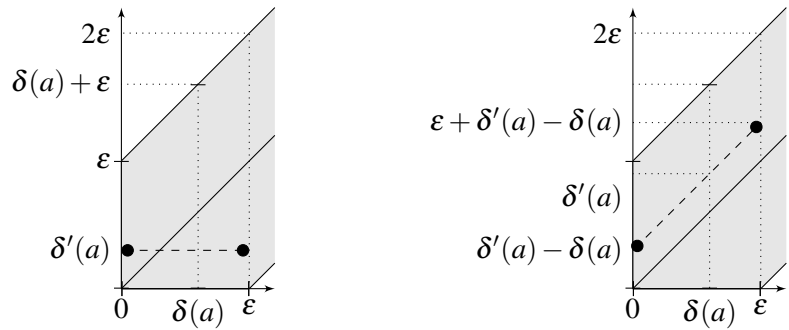
\begin{figure}[bt]
\centering
\begin{tikzpicture}[scale=1.7]
\begin{scope}
\draw[latex'-latex'] (1.2,0) -| (0,2.2);
\fill[black!10!white] (1,0) -- (1.2,.2) -- (1.2,2.2) -- (0,1) |- (1,0);
\draw[|-|] (0,0) -- (1,0);
\draw[|-|] (0,0) -- (0,1) node[left] {$\epsilon$};
\draw[|-|] (0,0) node[below] {$0$} -- (1,0) node[below] {$\epsilon$};
\draw (0,0) -- (1.2,1.2);
\draw[dotted] (1,0) -- (1,2) -- (0,2) node[left] {$2\epsilon$};
\draw (0,1) -- (1.2,2.2);
\draw (1,0) -- (1.2,.2);
\draw[|-|,dotted] (.6,0) node[below] {$\delta(a)$} -- (.6,1.6);
\draw[dotted] (.6,1.6) -- (0,1.6) node[left] {$\delta(a)+\epsilon$};
\draw[dashed,*-*] (0,.3) node[left] {$\delta'(a)$} -- (1,.3);
\end{scope}

\begin{scope}[xshift=4cm]
\draw[latex'-latex'] (1.2,0) -| (0,2.2);
\fill[black!10!white] (1,0) -- (1.2,.2) -- (1.2,2.2) -- (0,1) |- (1,0);
\draw[|-|] (0,0) -- (1,0);
\draw[|-|] (0,0) -- (0,1);
\draw[|-|] (0,0) node[below] {$0$} -- (1,0) node[below] {$\epsilon$};
\draw (0,0) -- (1.2,1.2);
\draw[dotted] (1,0) -- (1,2) -- (0,2) node[left] {$2\epsilon$};
\draw (0,1) -- (1.2,2.2);
\draw (1,0) -- (1.2,.2);
\draw[|-|,dotted] (.6,0) node[below] {$\delta(a)$} -- (.6,1.6);
\draw[dotted] (.6,1.6) -- (0,1.6); 
\draw[dashed,*-*] (0,.3) node[left] {$\delta'(a)-\delta(a)$} -- (1,1.3);
\draw[dotted] (1,1.3) -- (0,1.3)  node[left] {$\epsilon+\delta'(a)-\delta(a)$};
\draw[dotted] (.6,.9) -- (0,.9) node[below left] {$\delta'(a)$};
\end{scope}
\end{tikzpicture}
\caption{Intuition of the construction for $\delta(a)\leq \epsilon$: seeing
  $\delta(a)$ as a convex combination of~$0$ and~$\epsilon$, we~obtain
  $\delta'(a)$ as the same convex combination of the black dots.}

\label{fig-leqeps}
\end{figure}

For a $1$-player game $\calG$ for~$i$, we denote by $\dg{\calG}{\varepsilon}$
the previous transformation.
Our aim is to have a correspondence between (stochastic) moves of Player~$i$
from~$s$ in~$\calG$, and her move from the corresponding state~$\hat s$
in~$\dg\calG\varepsilon$. Our notion of correspondence is defined as follows:
\begin{tdefinition}
  Let $\sigma_i,\sigma'_i\in\mathbb{S}$ two strategies for the $1$-player game
  $\calG$ (played by $i$) such that $d(\sigma_i,\sigma'_i)\leq \varepsilon$,
  and $\hat{\sigma}$ a strategy profile in $\dg{\calG}{\varepsilon}$. We say
  that $(\sigma_i,\sigma'_i)$ corresponds to $\hat{\sigma}$ if the following
  holds for any history~$\hat{h}$ ending in state~$s$ of
  $\dg{\calG}{\varepsilon}$: 
  \[ 
  \Tab(s,\sigma'_i(\pi_\States(\hat
  h))) \equiv \widehat{\Tab}(s,\hat{\sigma}(\hat h)) 
  \]
  where $\pi_\States(h)$ the projection on the letters corresponding to the
  original states $\States$.
\end{tdefinition}

We now explicit explicit the purpose of the construction by establishing a
correspondence between strategies in the original game and strategies
in our $2$-player version. 
\begin{restatable}{tlemma}{lemmaxix}
  For any $\sigma_i$ strategy of~$\calG$, there exists a strategy
  $\hat{\sigma}_i$ 
  in $\dg{\calG}{\varepsilon}$ for player~$i$, such that, for any strategy
  $\sigma'_i$ of~$\calG$
  such that $d(\sigma_i,\sigma'_i)\leq \varepsilon$,
  there exists $\hat{\sigma}_{\hat{i}}$ such that
  $(\sigma_i,\sigma'_i)$ corresponds to $\hat{\sigma}$.

  Moreover, any pure memoryless strategy profile of~$\dg{\calG}{\varepsilon}$
  corresponds to some pair of strategies $(\sigma_i,\sigma'_i)$ in~$\calG$
  where $\sigma_i$ is pure memoryless and $\sigma'_i$ is stationary.
\end{restatable}

The constructed game is a turn-based stochastic game with a quantitative
terminal reachability objective, which can be interpreted as a special case of
limit-average objective. Hence, thanks to a result of~\cite{LL69}, such a game
is determined with pure memoryless optimal strategies for both players.

As a consequence of this construction, we can infer two possible characterizations
of imprecise deviations in stationary profiles:
\begin{tcorollary}
  \label{cor:dgvalue}
  The value of $\dg{\calG}{\varepsilon}$ at state $\widehat{s}$ can be
  expressed as the following quantity on game $\calG$:
  \[
  \sup_{\sigma\in M^\calG}\inf_{ \substack{\sigma'\in
      \mathbb{M}^\calG \\d(\sigma,\sigma')\leq \varepsilon} }
  \Esp^{\sigma'}(\phi_i \mid s)
  \]
\end{tcorollary}

\begin{tcorollary}
  \label{cor:stateq}
  Let $\sigma\in\mathbb{M}^\calG$ a stationary strategy profile in
  $\calG$.
  $\sigma$ is an equilibrium under $\epsilon$-imprecise deviations
  from state $s_0$, if and only if:
  \[\forall i\in\Agt.\
  \forall \sigma'_i \in M_i^\calG.\ \exists \sigma''_i \in
  \mathbb{M}_i^\calG\ \text{s.t.}\ \Esp^{\sigma[i/
    \sigma''_i]}(\phi_i\mid s_0) \leq \Esp^{\sigma}(\phi_i\mid s_0) \
  \text{and}\ d(\sigma'_i,\sigma''_i) \leq \epsilon\]
\end{tcorollary}

\begin{remark}
\label{rq:dgblow}
One can notice the construction of the deviation game and inferred results
have been applied to nodes with two allowed actions only. In fact, the same
reasoning can be generalized to an arbitrary number of allowed actions at the
expense of an exponential blowup: player $i$ has to announce simultaneously,
for each allowed action $a$, if its probability in the expected distribution
will be larger than $\varepsilon$ and/or smaller than $1-\varepsilon$.
Note however that for a given fixed bound on the number of actions, the size
of $\dg{\calG}{\varepsilon}$ is still polynomial.
\end{remark}

\subsection{Existence of equilibria under imprecise deviations}
We are now ready to prove Theorem~\ref{theo:existence}, that is, for
every $\epsilon>0$, the existence of a (stationary) equilibrium under
$\epsilon$-imprecise deviations from any state of stochastic
concurrent games with terminal-reward payoffs.

Our proof will rely on the following well-known fixed-point theorem,
that we will apply to a well-adapted sets of strategy profiles.
\begin{ttheorem}[\cite{kakutani1941}]
  \label{theo:kakutani}
  Let $X$ be a non-empty, compact and convex subset of some Euclidean
  space.  Let $f\colon X \rightarrow 2^X$ be a set-valued function on
  $X$ with a closed graph and the property that $f(x)$ is non-empty
  and convex for all $x \in X$. Then $f$ has a fixed point.
\end{ttheorem}

A Nash equilibrium~$\sigma$ can be characterized as containing, for
each player $i$, the best response~$\sigma_i$ to the strategies of the
other players. This can be expressed as a fixed point of the
\emph{best-response function} (\cite{nash50}).
Nevertheless, over game graphs, continuity of this best-response
function is not ensured. More precisely, the graph of the function is
not closed.  Let us consider for example game of Figure~\ref{exLF},
and write any stationary strategy profile $\sigma$ in this game as the
tuple $(\sigma_1(s \mid 1),\sigma_2(s \mid 2))$.  Then, if one player
decides to stop the game with any positive probability, the other
player has all incentive to purely continue the game, until reaching the
terminal state (with probability), hence: $\BR((x,y)) = \{(0,0)\}$ for
every $x,y>0$, where $\BR$ denotes the best-response function.
However, if the other player purely continues the game, the only way
to win some positive payoff $1/3$ is to play the stopping action with
positive probability, hence: $\BR((0,0)) = \{(x,y) \mid x,y>0\}$.  We
conclude that the graph is not closed, so Theorem~\ref{theo:kakutani}
cannot apply to the classical $\BR$ function. This is not surprising
as we know that Nash equilibria need not always exist (recall the
example given in Figure~\ref{exHoR}). On the other hand,
in~\cite{CJM04}, stationary $\epsilon$-Nash equilibria are
characterized as fixed points of the best-response function.

In the following we will see that the (standard) best-response
function will fit well in our setting.
  \begin{tdefinition}
    We consider $T\subseteq \mathbb{M}$ a subset of stationary
    strategy profiles.  Let $\BR_T\colon T \rightarrow 2^{T}$ with
	\[
        \BR_T(\sigma) = \left\{ \sigma'\in T~\middle|~ \forall
          i\in\Agt.\ \forall s\in\States.\ \sigma'_i \in
          \argmax_{\sigma''_i~\text{s.t.}~\sigma'[i/\sigma''_i]\in
            T}~ \Esp^{\sigma[i/\sigma''_i]}(\phi_i \mid s) \right\}\]
  \end{tdefinition}
  Note that $\BR_{\mathbb{M}}$ is the usual notion of best response
  function.

  \begin{tlemma}
    For every $0<\epsilon\leq \frac1{|\Act|}$ and $\calA$ cycle-free,
    $\BR_{\Delta_\epsilon}$ has a fixed point.
  \end{tlemma}

  \begin{proof}
    We apply Theorem~\ref{theo:kakutani}.
    \begin{itemize}
    \item 
    First notice that
    $T=\Delta_\epsilon$ can be viewed as a non-empty compact convex
    subset of $\mathbb{R}^N$ where $N=\Act\times\Agt\times \States$.
    Moreover, $T$ can be decomposed in a product
    of individual strategy sets for each player $T=T_1\times \hdots T_{|\Agt|}$
    where
    \[
      \forall i\in\Agt~
      T_i = \left\{ \sigma_i~\middle|~
        \forall (a,s)~(a,i,s)\in\exit(C)\Rightarrow
        \sigma_i(s)(a) \geq \epsilon
      \right\}
    \]
    Hence, for every $(\sigma,\sigma')\in T^2$, and $i\in \Agt$, we
    still have $\sigma[i/\sigma'_i] \in T$.  
    \item
    Let $k$ and $p$ be the constants appearing in the
    statement of Proposition~\ref{thmTerm}.
    For every $n \ge 0$, we define
    $g_{n}$ for the function assigning to every pair
    of strategy profiles $(\sigma,\sigma')\in T^2$
    the following vector value in $\mathbb{R}^{\Agt\times \States}$:
    \[
    \left( \sum_{j=0}^{k \cdot n} \sum_{f \in \Final}
      \pr^{\sigma[i/\sigma'_i]} \bigl( (\States \setminus \Final)^{j}
      \cdot f^\omega \mid s\bigr) \cdot \nu_i(f) \right)_{i\in\Agt,
      s\in \States}
    \]
    Then, we obviously see that for every $(i,s) \in \Agt \times
    \States$, $\lim_{n \to \infty} g_{n}(\sigma,\sigma')_{i,s} =
    \Esp^{\sigma[i/\sigma'_i]}(\phi_i \mid s)$.  Furthermore, as an
    application of Proposition~\ref{thmTerm}, we get:
    \[
     |\Esp^{\sigma[i/\sigma'_i]}(\phi_i \mid s) - g_{n}(\sigma,\sigma')_i |
     \le K \cdot p^n
    \]
    where $K=\max_{i\in \Agt,f \in \Final} |\nu_i(f)|$.  This implies
    that the above convergence is indeed uniform, and that $g_\infty:
    (\sigma,\sigma') \mapsto \left(\Esp^{\sigma[i/\sigma'_i]}(\phi_i
      \mid s)\right)_{i,s}$ is therefore continuous on $T^2$.  
    \item
        Let us now show that the graph of $\BR_T$ is closed. In order
        to do so, we consider a converging sequence of strategy
        profiles $(\sigma^k)_{k>0}$ with limit $\sigma^\infty$ and for
        each $k>0$, $\sigma'^k\in\BR_T(\sigma^k)$ converging to
        $\sigma'^\infty$.  We will prove that
        $\sigma'^\infty\in\BR_T(\sigma^\infty)$.  For a fixed
        $\sigma'$, we have $ \Esp^{\sigma^k[i/\sigma'_i]}(\phi_i \mid
        s) \leq \Esp^{\sigma^k[i/\sigma'^k_i]}(\phi_i \mid s)$, hence
        by continuity, $ \Esp^{\sigma^\infty[i/\sigma'_i]}(\phi_i \mid
        s) \leq \Esp^{\sigma^\infty[i/\sigma'^\infty_i]}(\phi_i \mid
        s)$.  
    \item
        It remains to show that $\BR_T(\sigma)$ is convex.  We fix
        $i\in \Agt$ and show that $(\BR_T(\sigma))_i$ is convex hence
        the result.  Let $0<\lambda<1$ and $\sigma',\sigma'' \in
        \BR_T(\sigma)$: this means that both vectors
        $(\Esp^{\sigma[i/\sigma'_i]}(\phi_i \mid s))_s$ and
        $(\Esp^{\sigma[i/\sigma''_i]} (\phi_i \mid s))_s$ are maximal,
        and equal to some vector $m_i$.  Indeed, if two different
        maximal vectors exists, we take the combined strategy that
        uses best action in each state, this new strategy is still in
        $T_i$.

        By convexity of $T=\Delta_\epsilon$, $\sigma^\lambda =
        \sigma[i/\lambda \cdot\sigma'_i + (1-\lambda) \cdot
        \sigma''_i] \in T$, so $\forall s$,
        $\pr^{\sigma^\lambda}(\States^* \Final \mid s) = 1$.  This
        implies that the payoff vector $(\Esp^{\sigma^\lambda}(\phi_i
        \mid s))_{s}$ is the \emph{unique} solution of the equation
        \[\left\{
          \begin{aligned}
            &\forall f\in \Final&~ \Esp^{\sigma^\lambda}(\phi_i \mid f) &= \nu_i(f) \\
            &\forall s\notin \Final&~ \Esp^{\sigma^\lambda}(\phi_i
            \mid s) &=
            \sum_{s'}\Tab(s,\sigma^\lambda(s))(s')\Esp^{\sigma^\lambda}(\phi_i
            \mid s')\\
            &&&=
            \sum_{s'}\left[
                \lambda\Tab(s,\sigma[i/\sigma'_i](s))
                +
                (1-\lambda)\Tab(s,\sigma[i/\sigma''_i](s))
            \right](s')\cdot
                \Esp^{\sigma^\lambda}(\phi_i\mid s')
    \end{aligned}
    \right.
    \]
    On the other hand, $m_i$ satisfies the following equation:
    \[\left\{
    \begin{aligned}
      &\forall f\in \Final&~ m_{i,f} &= \nu_i(f) \\
      &\forall s\notin \Final&~
        m_{i,s} &= 
        \sum_{s'}\Tab(s,\sigma[i/\sigma'_i](s))(s')m_{i,s'}
      =
        \sum_{s'}\Tab(s,\sigma[i/\sigma''_i](s))(s')m_{i,s'}\\
    \end{aligned}
    \right.
    \]
    We can check that
    $(\Esp^{\sigma^\lambda}(\phi_i \mid s))_{s\in \States} = m_i$
    is a valid solution, hence the actual value, so
    $\sigma^\lambda_i\in \BR_T(\sigma)_i$.
    \popQED
   \end{itemize}
  \end{proof}

  Thanks to Corollary~\ref{cor:stateq} (stationary deviations), and this
  fixed-point theorem, we infer the following proposition:
  \begin{tproposition}
    If $0< \epsilon \le \frac{1}{|\Act|}$ and $\calA$ is cycle-free,
    then there exists
    $\sigma \in \Delta_{\epsilon}$ fixed point of
    $\BR_{\Delta_{\epsilon}}$ which is an equilibrium under
    $\epsilon$-imprecise deviations from every state $s$ of~$\calG$.
  \end{tproposition}
  The general Theorem~\ref{theo:existence} follows immediately for any
  $\epsilon>0$ and any arena, thanks to Proposition~\ref{prop:cyclefree}.

\section{Computing stationary equilibria under imprecise
  deviations}

We describe a polynomial-space algorithm for computing stationary
equilibria under imprecise deviations for
  non-negative terminal reward games.
A~similar proof for Nash equilibria in turn-based stochastic games is
given in~\cite{UW11}. We~briefly describe the later proof, which will
help understanding our current encoding.

\looseness=-1
  The algorithm proceeds by encoding a Nash Equilibrium as an
  existential first-order formula over the reals, which satisfiability
  can be decided in \PSPACE. The formula
  quantifies over all stationary strategy profiles and payoffs at each
  state, and checks that:
  \begin{enumerate}
  \item the strategy profile~$\sigma$ under consideration is properly defined;
  \item the payoff in each state corresponds to the real payoff of the
    strategy profile;
  \item for any~$i$, Player~$i$ cannot benefit from deviating in
    $\comp{\calG}{\sigma}{i}$.
  \end{enumerate}
  These properties cannot, \emph{in general}, be expressed locally,
  but in the setting of~\cite{UW11}, one can first, non-deterministically, guess
  the support of the strategy.
  On~the one hand, this allows us to compute (in~linear time) the set of
  states from which $\Final$ is never reached. Those states have payoff~$0$
  for all agents, and the payoff in the other states (from which $\Final$ is
  reachable with some positive probability) can be expressed
  as a combination of the payoff values of the successor states and
  the (local) strategy profile.
  On~the other hand, we~can also compute (still in linear time) the set of
  states that are reachable from~$s_0$. It is easy to see that Player~$i$
  has an incentive to deviate if, and only~if, her payoff can be increased by
  deviating locally from such a reachable state. Hence we can express stability
  of the Nash Equilibrium as a (polynomial size) conjunction of inequalities.

  Another way of expressing this stability property is by saying that for any
  Player~$i$, $s_0$~should yield a payoff in the equilibrium that is larger
  than the optimal value~$v_i(s_0)$ in the Markov decision process
  representing the possible deviations of Player~$i$, namely
  $\comp{\calG}{\sigma}{i}$. Since the initial guess can be done in \NPSPACE{}
  and the generated formula is of polynomial size, the whole algorithms runs
  in \PSPACE.

  In the case of equilibria under $\epsilon$-imprecise deviations, we apply a
similar technique but deviations are now to be considered as strategies for
Player~$i$ in $\dg{\comp{\calG}{\sigma}{i}}{\varepsilon}$ against the worst
strategies of Player~$\widehat{i}$.
In~fact,
we want to check that $s_0$ has a
payoff (in the equilibrium) larger for player $i$ than the maximal value she
could get by imprecisely deviating. Thanks to corollary~\ref{cor:dgvalue},
this optimal value is the same as in
$\calG_i=\dg{\comp{\calG}{\sigma}{i}}{\varepsilon}$, denoted by $v_{\varepsilon,i}(s)$.
In~order to compute these values for each
game~$\calG_i$, we~non-deterministically compute optimal strategies for
players~$i$ and~$\hat{i}$. These strategies can be supposed to be pure
memoryless. In~order to do~so, we~first guess a strategy for Player~$i$ in the
game game $\dg{\comp{\calG}{\sigma}{i}}{\varepsilon}$. Without knowing the
exact probability values of this game (which depends on~$\sigma$), we can
still derive its structure since the support is known, thus we can compute the
set of states for which Player~$\hat{i}$ can totally spoil $i$'s payoff, that~is,
enforce a non-terminating run; such a run has payoff~$0$, which is
optimal for
Player~$\hat{i}$. We~later guess a pure memoryless strategy for
Player~$\hat{i}$ keeping in mind that $\hat{i}$ has to play such a cycling
strategy from any state where she is able~to. From the other states, for which
Player~$i$ can still ensure positive probability to terminate, the value of the
game can again be expressed locally as a combination of the guessed strategy
profile and the values of the successor states. As~for the previous algorithm
for Nash Equilibrium in~$\calG$, the optimality of both strategies can be
expressed as stability by local deviations.
Finally, stability by imprecise deviations in~$\calG$ consists in coding
the fact that payoff in~$\calG$ for Player~$i$ should be larger than the optimal
value $v_{\varepsilon,i}(s_0)$.

We now make precise the result and the algorithm.

\begin{restatable}{ttheorem}{thmxxvii}
  Let $k>0$.
  Let $\calG = \tuple{ \calA,\phi_\nu }$ be
    a stochastic concurrent game with non-negative terminal
    rewards with $|\Act|\leq k$.
    Let $s_0 \in \States$ and $\epsilon>0$. For every $i \in
    \Agt$, we fix $x_i, y_i \in \mathbb{R}_+$ two real numbers. We can
    decide in \PSPACE{} whether there is a stationary equilibrium under
    $\epsilon$-imprecise deviations $\sigma$ from $s_0$, such that for
    every $i \in \Agt$,
    $x_i \le \Esp^\sigma(\phi_i \mid s_0) \le y_i$.

    \end{restatable}

\begin{remark}
  The previous theorem can be applied to compute \emph{some} equilibria in the case of
  negative payoffs by considering the new payoff function $\nu' = \nu-\min\nu\geq 0$.
  However, $\phi'=\phi_\nu-\min\nu$ and $\phi_{\nu'}$ coincide only on runs that
  reach a final state since $\phi'$ assigns positive value $-\min\nu$ to
  non-terminating runs. A possible work-around is to first compute the cycle-free
  arena $\widetilde{\calA}$ and exiting conditions $\Delta_\epsilon$, which size
  is bounded by the number of pairs $(a,i,s)\in \Act\times\Agt\times\States$.
  Then we can apply the previous theorem on game
  $\langle \widetilde{\calA}, \phi_{\nu'}\rangle$ with the extra
  formula $\sigma\in\Delta_\epsilon$.
  Thanks to this last constraint, we ensure that the run always terminates,
  thus the payoff functions coincide. Finally we conclude the computation by
  applying proposition~\ref{prop:cyclefree} to get back an equilibrium on $\calG$.
\end{remark}

\bibliographystyle{eptcs}
\bibliography{bib-reduite}

\begin{thebibliography}{10}
\providecommand{\bibitemdeclare}[2]{}
\providecommand{\surnamestart}{}
\providecommand{\surnameend}{}
\providecommand{\urlprefix}{Available at }
\providecommand{\url}[1]{\texttt{#1}}
\providecommand{\href}[2]{\texttt{#2}}
\providecommand{\urlalt}[2]{\href{#1}{#2}}
\providecommand{\doi}[1]{doi:\urlalt{http://dx.doi.org/#1}{#1}}
\providecommand{\bibinfo}[2]{#2}

\bibitemdeclare{article}{AT12}
\bibitem{AT12}
\bibinfo{author}{D.~\surnamestart Auger\surnameend} \&
  \bibinfo{author}{O.~\surnamestart Teyraud\surnameend} (\bibinfo{year}{2012}):
  \emph{\bibinfo{title}{The Frontier of Decidability in Partially Observable
  Recursive Games}}.
\newblock {\sl \bibinfo{journal}{Int. Journal of Foundations of Computer
  Science}} \bibinfo{volume}{23}(\bibinfo{number}{7}), pp.
  \bibinfo{pages}{1439--1450}, \doi{10.1142/S0129054112400576}.

\bibitemdeclare{inproceedings}{BBMU11}
\bibitem{BBMU11}
\bibinfo{author}{P.~\surnamestart Bouyer\surnameend},
  \bibinfo{author}{R.~\surnamestart Brenguier\surnameend},
  \bibinfo{author}{N.~\surnamestart Markey\surnameend} \&
  \bibinfo{author}{M.~\surnamestart Ummels\surnameend} (\bibinfo{year}{2011}):
  \emph{\bibinfo{title}{{N}ash Equilibria in Concurrent Games with {B}{\"u}chi
  Objectives}}.
\newblock In: {\sl \bibinfo{booktitle}{Proc. 30th Conf. on Foundations of
  Software Technology and Theoretical Computer Science (FSTTCS'11)}}, {\sl
  \bibinfo{series}{LIPIcs}}~\bibinfo{volume}{13},
  \bibinfo{publisher}{Leibniz-Zentrum f{\"u}r Informatik}, pp.
  \bibinfo{pages}{375--386}, \doi{10.4230/LIPIcs.FSTTCS.2011.375}.

\bibitemdeclare{article}{BBMU15}
\bibitem{BBMU15}
\bibinfo{author}{P.~\surnamestart Bouyer\surnameend},
  \bibinfo{author}{R.~\surnamestart Brenguier\surnameend},
  \bibinfo{author}{N.~\surnamestart Markey\surnameend} \&
  \bibinfo{author}{M.~\surnamestart Ummels\surnameend} (\bibinfo{year}{2015}):
  \emph{\bibinfo{title}{Pure {N}ash Equilibria in Concurrent Games}}.
\newblock {\sl \bibinfo{journal}{Logical Methods in Computer Science}}
  \bibinfo{volume}{11}(\bibinfo{number}{2:9}), \doi{10.2168/LMCS-11(2:9)2015}.

\bibitemdeclare{inproceedings}{BMS14}
\bibitem{BMS14}
\bibinfo{author}{P.~\surnamestart Bouyer\surnameend},
  \bibinfo{author}{N.~\surnamestart Markey\surnameend} \&
  \bibinfo{author}{D.~\surnamestart Stan\surnameend} (\bibinfo{year}{2014}):
  \emph{\bibinfo{title}{Mixed {N}ash Equilibria in Concurrent Games}}.
\newblock In: {\sl \bibinfo{booktitle}{Proc. 33rd Conf. on Foundations of
  Software Technology and Theoretical Computer Science (FSTTCS'14)}}, {\sl
  \bibinfo{series}{LIPIcs}}~\bibinfo{volume}{29},
  \bibinfo{publisher}{Leibniz-Zentrum f{\"u}r Informatik}, pp.
  \bibinfo{pages}{351--363}, \doi{10.4230/LIPIcs.FSTTCS.2014.351}.

\bibitemdeclare{article}{CD14}
\bibitem{CD14}
\bibinfo{author}{K.~\surnamestart Chatterjee\surnameend} \&
  \bibinfo{author}{L.~\surnamestart Doyen\surnameend} (\bibinfo{year}{2014}):
  \emph{\bibinfo{title}{Partial-Observation Stochastic Games: How to Win when
  Belief Fails}}.
\newblock {\sl \bibinfo{journal}{ACM Transactions on Computational Logic}}
  \bibinfo{volume}{15}(\bibinfo{number}{2:16}), \doi{10.1145/2579821}.

\bibitemdeclare{inproceedings}{CJM04}
\bibitem{CJM04}
\bibinfo{author}{K.~\surnamestart Chatterjee\surnameend},
  \bibinfo{author}{M.~\surnamestart Jurdzi{\'n}ski\surnameend} \&
  \bibinfo{author}{R.~\surnamestart Majumdar\surnameend}
  (\bibinfo{year}{2004}): \emph{\bibinfo{title}{On {N}ash Equilibria in
  Stochastic Games}}.
\newblock In: {\sl \bibinfo{booktitle}{Proc. 18th Int. Workshop on Computer
  Science Logic (CSL'04)}}, {\sl \bibinfo{series}{LNCS}}
  \bibinfo{volume}{3210}, \bibinfo{publisher}{Springer}, pp.
  \bibinfo{pages}{26--40}, \doi{10.1007/978-3-540-30124-0\_6}.

\bibitemdeclare{inproceedings}{CKPS10}
\bibitem{CKPS10}
\bibinfo{author}{T.~\surnamestart Chen\surnameend},
  \bibinfo{author}{M.~\surnamestart Kwiatkowska\surnameend},
  \bibinfo{author}{D.~\surnamestart Parker\surnameend} \&
  \bibinfo{author}{A.~\surnamestart Simaitis\surnameend}
  (\bibinfo{year}{2011}): \emph{\bibinfo{title}{Verifying Team Formation
  Protocols with Probabilistic Model Checking}}.
\newblock In: {\sl \bibinfo{booktitle}{Proc. 12th Int. Workshop on
  Computational Logic in Multi-Agent Systems (CLIMA'11)}}, {\sl
  \bibinfo{series}{LNAI}} \bibinfo{volume}{6814},
  \bibinfo{publisher}{Springer}, pp. \bibinfo{pages}{190--207},
  \doi{10.1007/978-3-642-22359-4\_14}.

\bibitemdeclare{incollection}{WTMKD15}
\bibitem{WTMKD15}
\bibinfo{author}{A.~\surnamestart Das\surnameend},
  \bibinfo{author}{S.~\surnamestart Krishna\surnameend},
  \bibinfo{author}{L.~\surnamestart Manasa\surnameend},
  \bibinfo{author}{A.~\surnamestart Trivedi\surnameend} \&
  \bibinfo{author}{D.~\surnamestart Wojtczak\surnameend}
  (\bibinfo{year}{2015}): \emph{\bibinfo{title}{On Pure {N}ash Equilibria in
  Stochastic Games}}.
\newblock In: {\sl \bibinfo{booktitle}{Theory and Applications of Models of
  Computation}}, {\sl \bibinfo{series}{LNCS}} \bibinfo{volume}{9076},
  \bibinfo{publisher}{Springer}, pp. \bibinfo{pages}{359--371},
  \doi{10.1007/978-3-319-17142-5\_31}.

\bibitemdeclare{inproceedings}{henzinger05}
\bibitem{henzinger05}
\bibinfo{author}{Thomas~A. \surnamestart Henzinger\surnameend}
  (\bibinfo{year}{2005}): \emph{\bibinfo{title}{Games in System Design and
  Verification}}.
\newblock In: {\sl \bibinfo{booktitle}{Proceedings of the 10th Conference on
  Theoretical Aspects of Rationality and Knowledge}}, \bibinfo{series}{TARK
  '05}, \bibinfo{publisher}{National University of Singapore},
  \bibinfo{address}{Singapore, Singapore}, pp. \bibinfo{pages}{1--4}.
\newblock \urlprefix\url{http://doi.acm.org/10.1145/1089933.1089935}.

\bibitemdeclare{article}{kakutani1941}
\bibitem{kakutani1941}
\bibinfo{author}{S.~\surnamestart Kakutani\surnameend} (\bibinfo{year}{1941}):
  \emph{\bibinfo{title}{A generalization of Brouwer's fixed point theorem}}.
\newblock {\sl \bibinfo{journal}{Duke Mathemastical Journal}}
  \bibinfo{volume}{8}(\bibinfo{number}{3}), pp. \bibinfo{pages}{457--459},
  \doi{10.1215/S0012-7094-41-00838-4}.

\bibitemdeclare{article}{nash50}
\bibitem{nash50}
\bibinfo{author}{J.F. \surnamestart Nash\surnameend} (\bibinfo{year}{1950}):
  \emph{\bibinfo{title}{Equilibrium Points in $n$-Person Games}}.
\newblock {\sl \bibinfo{journal}{Proceedings of the National Academy of
  Sciences of the United States of America}}
  \bibinfo{volume}{36}(\bibinfo{number}{1}), pp. \bibinfo{pages}{48--49},
  \doi{10.1073/pnas.36.1.48}.

\bibitemdeclare{article}{SS01}
\bibitem{SS01}
\bibinfo{author}{P.~\surnamestart Secchi\surnameend} \& \bibinfo{author}{W.D.
  \surnamestart Sudderth\surnameend} (\bibinfo{year}{2001}):
  \emph{\bibinfo{title}{Stay-in-a-Set Games}}.
\newblock {\sl \bibinfo{journal}{Int. Journal of Game Theory}}
  \bibinfo{volume}{30}, pp. \bibinfo{pages}{479--490},
  \doi{10.1007/s001820200092}.

\bibitemdeclare{article}{selten65}
\bibitem{selten65}
\bibinfo{author}{R.~\surnamestart Selten\surnameend} (\bibinfo{year}{1965}):
  \emph{\bibinfo{title}{Spieltheoretische {B}ehandlung eines {O}ligopolmodells
  mit {N}achfragetr\"agheit}}.
\newblock {\sl \bibinfo{journal}{Zeitschrift f\"ur die gesamte
  Staatswissenschaft}} \bibinfo{volume}{121}(\bibinfo{number}{2}), pp.
  \bibinfo{pages}{301--324 and 667--689}.
\newblock \urlprefix\url{http://www.jstor.org/stable/40748884}.

\bibitemdeclare{article}{selten75}
\bibitem{selten75}
\bibinfo{author}{R.~\surnamestart Selten\surnameend} (\bibinfo{year}{1975}):
  \emph{\bibinfo{title}{A reexamination of the perfectness concept for
  equilibrium points in extensive games}}.
\newblock {\sl \bibinfo{journal}{Int. Journal of Game Theory}}
  \bibinfo{volume}{4}, pp. \bibinfo{pages}{25--55}, \doi{10.1007/BF01766400}.

\bibitemdeclare{inproceedings}{thomas02}
\bibitem{thomas02}
\bibinfo{author}{W.~\surnamestart {\relax Th}omas\surnameend}
  (\bibinfo{year}{2002}): \emph{\bibinfo{title}{Infinite Games and
  Verification}}.
\newblock In: {\sl \bibinfo{booktitle}{Proc. 14th Int. Conf. on Computer Aided
  Verification (CAV'02)}}, {\sl \bibinfo{series}{LNCS}} \bibinfo{volume}{2404},
  \bibinfo{publisher}{Springer}, pp. \bibinfo{pages}{58--64},
  \doi{10.1007/3-540-45657-0\_5}.
\newblock \bibinfo{note}{Invited Tutorial}.

\bibitemdeclare{article}{LL69}
\bibitem{LL69}
\bibinfo{author}{S.A.~Lippman \surnamestart T.M.~Liggett\surnameend}
  (\bibinfo{year}{1969}): \emph{\bibinfo{title}{Short Notes: {S}tochastic Games
  With Perfect Information and Time Average Payoff}}.
\newblock {\sl \bibinfo{journal}{SIAM Review}}
  \bibinfo{volume}{11}(\bibinfo{number}{4}), pp. \bibinfo{pages}{604--607},
  \doi{10.1137/1011093}.

\bibitemdeclare{inproceedings}{ummels08}
\bibitem{ummels08}
\bibinfo{author}{M.~\surnamestart Ummels\surnameend} (\bibinfo{year}{2008}):
  \emph{\bibinfo{title}{The Complexity of {N}ash Equilibria in Infinite
  Multiplayer Games}}.
\newblock In: {\sl \bibinfo{booktitle}{Proc. 11th Int. Conf. on Foundations of
  Software Science and Computation Structures (FoSSaCS'08)}}, {\sl
  \bibinfo{series}{LNCS}} \bibinfo{volume}{4962},
  \bibinfo{publisher}{Springer}, pp. \bibinfo{pages}{20--34},
  \doi{10.1007/978-3-540-78499-9\_3}.

\bibitemdeclare{inproceedings}{UW11a}
\bibitem{UW11a}
\bibinfo{author}{M.~\surnamestart Ummels\surnameend} \&
  \bibinfo{author}{D.~\surnamestart Wojtczak\surnameend}
  (\bibinfo{year}{2011}): \emph{\bibinfo{title}{The Complexity of {N}ash
  Equilibria in Limit-Average Games}}.
\newblock In: {\sl \bibinfo{booktitle}{Proc. 22nd Int. Conf. on Concurrency
  Theory (CONCUR'11)}}, {\sl \bibinfo{series}{LNCS}} \bibinfo{volume}{6901},
  \bibinfo{publisher}{Springer}, pp. \bibinfo{pages}{482--496},
  \doi{10.1007/978-3-642-23217-6\_32}.

\bibitemdeclare{article}{UW11}
\bibitem{UW11}
\bibinfo{author}{M.~\surnamestart Ummels\surnameend} \&
  \bibinfo{author}{D.~\surnamestart Wojtczak\surnameend}
  (\bibinfo{year}{2011}): \emph{\bibinfo{title}{The Complexity of {N}ash
  Equilibria in Stochastic Multiplayer Games}}.
\newblock {\sl \bibinfo{journal}{Logical Methods in Comp. Science}}
  \bibinfo{volume}{7}(\bibinfo{number}{3}), \doi{10.2168/LMCS-7(3:20)2011}.

\end{thebibliography}

\end{document}